\documentclass[a4paper,11pt]{article}

\usepackage{geometry}
\usepackage{amsfonts}
\usepackage{amsmath}
\usepackage{amsthm}
\usepackage{hyperref}
\usepackage{bm}
\usepackage{xcolor}
\usepackage{cite}
\usepackage{authblk}
\usepackage{mathrsfs}
\usepackage{verbatim}
\allowdisplaybreaks

\title{On a two-component Camassa-Holm equation}

\author{Zixin Zhang}
\author{Q. P. Liu\thanks{qpl@cumtb.edu.cn}}
\affil{Department of Mathematics, \\
China University of Mining and Technology,\\ Beijing 100083, People's Republic of China}

\date{}
 
\begin{document}
\maketitle

\begin{abstract}
A two-component generalization of the Camassa-Holm equation and its reduction proposed recently by Xue, Du and Geng [Appl. Math. Lett. {\bf 146} (2023) 108795] are studied. For this two-component equation, its missing bi-Hamiltonian structure is constructed and a Miura transformation is introduced so that it may be regarded as a modification of the very first two-component Camassa-Holm equation. 
 Using a proper reciprocal transformation, a particular reduction of this two-component equation, which admits $N-$ peakon solution,  is brought  to the celebrated Burgers equation.
\end{abstract}

\section{Introduction}
The following equation, known as the Camassa-Holm (CH) equation, 
\begin{equation}\label{eq0-0}
m_t+um_x+2u_xm=0, \quad m=u-u_{xx}
\end{equation}
has been one of the most important integrable systems since the work of Camassa and Holm \cite{CH1993} and has been studied extensively (see \cite{Fuchss96,Schiff1998,BSS2000,Johnson2002,LZ2004,Parker2005,CGI2006,Monvel2007,CL2009,Feng10,Xia16,Rasin-Schiff2017,Li-Tian,Chang2024} and the references therein). In addition to the integrability, the most remarkable property enjoyed by the CH equation is that it possesses the peaked soliton or peakon solutions.  

Finding equations of CH type has been a hot topic in the soliton theory and consequently many new   equations have been constructed. In the single component case, the DP, Novikov and modified CH equations are  the most well-known ones \cite{DP, Novikov,olver1996,Qiao}. In the two or multi-component case, a number of  equations have been presented and the first example may be the two-component CH equation \cite{olver1996,clz2006}, which has studied from various viewpoints \cite{Wu2006,HI2011,Matsuno2017,Wang2020,Qu2024}. Recently, by means of the perturbative symmetry approach, Hone, Novikov and Wang classified  two-component CH equations and provided a list of such equations \cite{hnw2017}.

Very recently, based on a $2\times 2$ matrix spectral problem, Xue, Du and Geng proposed the following equation
\begin{align}\label{eq01}
    \left\{     
        \begin{aligned}
       n_{ t}&=4 [ n \left(v_{ x}+2 \beta v\right)]_{x},\\
    n&=4\beta^{2} v-v_{ x x},
\end{aligned}\right.
\end{align} 
and showed that it possesses $N-$peakon solution \cite{xdg2023}. Eq. \eqref{eq01} is a particular reduction of a two-component system, which is given by 
\begin{align}\label{eq1}
    \left\{     
        \begin{aligned}
    m_{ t}&=4(\partial_x m+m \partial_x)\left(v_{ x}+2 \beta v-u\right)+4 n  \left(u_{x}-2 \beta u-4 \alpha v \right)_x, \\
   n_{ t}&=4 \left[ n \left(v_{ x}+2 \beta v-u\right)\right]_x ,\\
   m&=u_{x x}-4\left(\alpha+\beta^{2}\right) u, \\
    n&=4\left(\alpha+\beta^{2}\right) v-v_{ x x},
\end{aligned}
\right.
\end{align} 
where $\alpha$ and $\beta$ are two constants, and $u=u(x,t), v=v(x,t)$. Apart from the $N-$peakon solution, those authors also constructed its infinitely many conservation laws.  For convenience, Eq. \eqref{eq1} will be referred to as the Xue-Du-Geng (XDG) equation in the ensuing discussion.

In this Letter, we aim to offer a better understanding of \eqref{eq01} and \eqref{eq1}. We shall construct two local Hamiltonian operators and show that Eq. \eqref{eq1} is a genuine bi-Hamiltonian system. We also introduce a Miura type transformation for Eq. \eqref{eq1} and demonstrate that this equation qualifies a modification of the 
well-known two-component CH (2-CH) equation \cite{clz2006, olver1996}. Furthermore, we explain the relationship between Eq.  \eqref{eq01} and the Burgers equation.  

This Letter is organized as follows. In the next section, Section 2, we apply the gradient spectral method and
calculate the missing bi-Hamiltonian operators for Eq. \eqref{eq1}. In Section 3, we build a transformation of Miura type which relates Eq. \eqref{eq1} to the 2-CH equation. In the last section, by means of reciprocal transformation, a link between Eq. \eqref{eq01} and the Burgers equation is established. 

\section{Bi-Hamiltonian structure}
Bi-Hamiltonian structure is an important ingredient for a given integrable system.
We now show that, by applying the  spectral gradient method, we may obtain the  missing bi-Hamiltonian structure for the XDG equation \eqref{eq1}. 
 
Let us begin with the spatial part of spectral problem of \eqref{eq1}
\begin{equation}\label{eq6.1}
  \binom{\phi_{1}}{\phi_{2}}_{x}=\left(\begin{array}{cc}
        \beta+\lambda n & 1 \\
        \alpha+\lambda m & -\beta-\lambda n
        \end{array}\right)\binom{\phi_{1}}{\phi_{2}}.
\end{equation}
To proceed, we need the  adjoint problem which is given by
\begin{equation}\label{eq6}
\left(\psi_{1}, \psi_{2}\right)_{x}=-\left(\psi_{1}, \psi_{2}\right)\left(\begin{array}{cc}
    \beta+\lambda n & 1 \\
        \alpha+\lambda m & -\beta-\lambda n
    \end{array}\right).
\end{equation}

We now calculate the variational derivatives of $\lambda$ with respect to $m$ and $n$, which will be denoted by $\lambda_m$ and $\lambda_n$, respectively. 
The directional derivative of $\left(\phi_{1}, \phi_{2}\right)$ 
in the direction $m+\epsilon M$ is given by
\begin{equation}\label{eq5}
    \binom{\phi_{1}^{\prime}[M]}{\phi_{2}^{\prime}[M]}_{x}=\left(\begin{array}{cc}
        \beta+\lambda n & 1 \\
        \alpha+\lambda m & -\beta-\lambda n
        \end{array}\right)\binom{\phi_{1}^{\prime}[M]}{\phi_{2}^{\prime}[M]}+\left(\begin{array}{cc}
        \left\langle M, \lambda_{m}\right\rangle n& 0 \\
        \left\langle M, \lambda_{m}\right\rangle m+\lambda M & -\left\langle M, \lambda_{m}\right\rangle n
        \end{array}\right)\binom{\phi_{1}}{\phi_{2}},               
\end{equation}
 where $\langle\cdot, \cdot\rangle $ signifies the pairing between tangent vectors and cotangent vectors, and is defined by  
\begin{equation*}\langle\rho, \sigma\rangle=\int_{-\infty}^{\infty} \rho(x) \sigma(x) \mathrm{d} x.\end{equation*} 
Now multiplying Eq.(\ref{eq5}) from left by $\left(\psi_{1}, \psi_{2}\right) $ and integrating it  from $-\infty$ to $\infty$, we obtain
\begin{equation}\label{eq7}
\begin{aligned}
    \int_{-\infty}^{\infty}\left(\psi_{1}, \psi_{2}\right) \binom{\phi_{1}^{\prime}[M]}{\phi_{2}^{\prime}[M]}_{x} \mathrm{d} x= & \int_{-\infty}^{\infty}\left(\psi_{1}, \psi_{2}\right)\left(\begin{array}{cc}
        \beta+\lambda n & 1 \\
        \alpha+\lambda m & -\beta-\lambda n
    \end{array}\right)\binom{\phi_{1}^{\prime}[M]}{\phi_{2}^{\prime}[M]} \mathrm{d} x \\
    & +\int_{-\infty}^{\infty}\left(\psi_{1} ,\psi_{2}\right)\left(\begin{array}{cc}
        \left\langle M, \lambda_{m}\right\rangle n& 0 \\
        \left\langle M, \lambda_{m}\right\rangle m+\lambda M & -\left\langle M, \lambda_{m}\right\rangle n
    \end{array}\right)\binom{\phi_{1}}{\phi_{2}} \mathrm{~d} x .
    \end{aligned}
\end{equation}
Taking Eq.(\ref{eq6}) into consideration, we have 
\begin{equation*}
        0=\int_{-\infty}^{\infty}\left(\psi_{1} ,\psi_{2}\right)\left(\begin{array}{cc}
            \left\langle M, \lambda_{m}\right\rangle n& 0 \\
            \left\langle M, \lambda_{m}\right\rangle m+\lambda M & -\left\langle M, \lambda_{m}\right\rangle n
        \end{array}\right)\binom{\phi_{1}}{\phi_{2}} \mathrm{~d} x, \\
    \end{equation*}
    i.e.\begin{equation}\label{eq8}
        0=\left\langle M, \lambda_{m}\right\rangle \int_{-\infty}^{\infty}\left(n \psi_{1} \phi_{1}-n \psi_{2} \phi_{2}+m \psi_{2} \phi_{1}\right) \mathrm{d} x+\lambda \int_{-\infty}^{\infty} M \psi_{2} \phi_{1} \mathrm{~d} x.
\end{equation}
As Eq.(\ref{eq8}) holds for arbitrary $M$,   
we derive the gradient of $\lambda$ with respect to $m$
\begin{equation*}
    \lambda_{m}=-\frac{\lambda \psi_{2} \phi_{1}}{\int_{-\infty}^{\infty}\left(n \psi_{1} \phi_{1}-n \psi_{2} \phi_{2}+m \psi_{2} \phi_{1}\right) \mathrm{d} x}.
\end{equation*}

Similarly, the gradient of $\lambda$ with respect to $n$ may be obtained as  
\begin{equation*}
        \lambda_{n}=-\frac{\lambda (\psi_{1} \phi_{1}-\psi_{2} \phi_{2})}{\int_{-\infty}^{\infty}\left(n \psi_{1} \phi_{1}-n \psi_{2} \phi_{2}+m \psi_{2} \phi_{1}\right) \mathrm{d} x}.
\end{equation*}
Without loss of generality, the constant coefficient can be omitted, and we have
\begin{equation*}
\binom{\lambda_{m}}{\lambda_{n}} \propto\binom{\psi_{2} \phi_{1}}{\psi_{1} \phi_{1}-\psi_{2} \phi_{2}}
\end{equation*}which satisfies
\newtheorem{prop}{Proposition}
\begin{prop}
The gradients of  the spectral parameter of the spectral problem \eqref{eq6.1} satisfy  
\begin{equation*}
\mathcal{E}_{2}\binom{\lambda_{m}}{\lambda_{n}}=2\lambda \mathcal{E}_{1} \binom{\lambda_{m}}{\lambda_{n}},
\end{equation*}
where
\begin{equation*}
    \mathcal{E}_{1} =\left(\begin{array}{cc}
       m \partial_{x}+\partial_{x} m& n \partial_{x}\\
       \partial_{x} n&0
       \end{array}\right),\quad \mathcal{E}_{2}= \left(\begin{array}{cc}
             -4\alpha \partial_{x}&-\partial_{x}^{2}-2\beta \partial_{x}\\
        \partial_{x}^{2}-2\beta \partial_{x}& \partial_{x}  
     \end{array}\right).
\end{equation*} 

   \end{prop}
\begin{proof}
    From Eqs.(\ref{eq6.1}) and (\ref{eq6}), we easily obtain
        \begin{align}
    	\left(\psi_{2} \phi_{1}\right)_{x}&=2 \left(\beta+\lambda n\right) \psi_{2} \phi_{1}-\left(\psi_{1} \phi_{1}-\psi_{2} \phi_{2}\right), \label{eq7-000}\\
	           \left(\psi_{1} \phi_{2}\right)_{x}&=\left(\alpha +\lambda m\right) \left(\psi_{1} \phi_{1}-\psi_{2} \phi_{2}\right)-2 \left(\beta+\lambda n\right) \psi_{1} \phi_{2},\label{eq7-00}\\
\label{eq7-0}            \left(\psi_{1} \phi_{1}-\psi_{2} \phi_{2}\right)_{x}&=2 \psi_{1} \phi_{2}-2\left(\alpha+\lambda m\right) \psi_{2} \phi_{1}.
        \end{align}
Also, we notice that
    \begin{equation*}\label{eq9}
        \begin{aligned}
            \left(\psi_{1} \phi_{1}-\psi_{2} \phi_{2}\right)_{x x} =&-\big(2 \lambda m_{x}+4\left(\alpha+\lambda m\right) \left(\beta+\lambda n\right)\big)\psi_{2} \phi_{1}+4\left(\alpha+\lambda m\right) \left(\psi_{1} \phi_{1} -\psi_{2} \phi_{2}\right)\\&-4\left(\beta+\lambda n\right) \psi_{1} \phi_{2},
        \end{aligned}
    \end{equation*}
which, after eliminating $\psi_{1} \phi_{2}$ by \eqref{eq7-0} and taking \eqref{eq7-000} into consideration, may be reformulated as 
\begin{equation}\label{11.1}
    \left(\psi_{1} \phi_{1}-\psi_{2} \phi_{2}\right)_{x x}+2 \left(\beta+\lambda n\right) \left(\psi_{1} \phi_{1}-\psi_{2} \phi_{2}\right)_{x}=-4\alpha \left(\psi_{2} \phi_{1}\right)_{x}-2 \lambda \left(m \partial_{x}+\partial_{x} m\right) \psi_{2} \phi_{1}.
\end{equation} 
Similarly, we find
\begin{align*}\label{eq10}
    \left(\psi_{2} \phi_{1}\right)_{x x}=2 \lambda n_{x} \psi_{2} \phi_{1}+\left(4\left(\beta+\lambda n \right)^{2}+2\left(\alpha+\lambda m\right)\right)\psi_{2} \phi_{1}-2\left(\beta+\lambda n\right) \left(\psi_{1} \phi_{1}-\psi_{2} \phi_{2}\right)-2 \psi_{1} \phi_{2},
\end{align*}
which, with the aid of Eqs. \eqref{eq7-000}-\eqref{eq7-0}, may be rewritten as
   \begin{equation}\label{11.2}
     \begin{aligned}
         \left(\psi_{2} \phi_{1}\right)_{x x}+\left(\psi_{1} \phi_{1}-\psi_{2} \phi_{2}\right)_{x}=2\lambda \left( n  \psi_{2} \phi_{1}\right)_{x}+2\beta \left(\psi_{2} \phi_{1}\right)_{x}.
     \end{aligned}
   \end{equation}
Thus, writing Eqs.(\ref{11.1}) and (\ref{11.2}) in matrix form, one can get the following eigenvalue problem of the spectral gradients
\begin{equation*}\label{eq11-3}
    \left(\begin{array}{cc}
       -4\alpha \partial_{x}&-\partial_{x}^{2}-2\beta \partial_{x}\\
        \partial_{x}^{2}-2\beta \partial_{x}& \partial_{x}     \end{array}\right)\binom{\lambda_{m}}{\lambda_{n}}=2 \lambda\left(\begin{array}{cc}
       m \partial_{x}+\partial_{x} m& n \partial_{x}\\
       \partial_{x} n&0
      \end{array}\right)\binom{\lambda_{m}}{\lambda_{n}},      
        \end{equation*}
and the proposition is proved.
\end{proof}     
%
The operator $\mathcal{E}_{2}$ is obviously a Hamiltonian operator. It is not difficult to prove that the operator $\mathcal{E}_{1}$ also is a Hamiltonian operator. Rather than offering a direct proof, we will show in the next section that these operators may be related to the Hamiltonian operators of the two-component CH equation. This fact implies that  $\mathcal{E}_{1}$ and  $\mathcal{E}_{2}$ constitute a compatible Hamiltonian pair.

Accordingly, we may generate a hierarchy from the spectral problem \eqref{eq6.1} and in particular Eq.\eqref{eq1} is a bi-Hamiltonian system 
\begin{equation*}
   \binom{m}{n}_{t} =\mathcal{E}_1 \binom{\frac{\delta{{H}}_{1}}{\delta{m}}}{\frac{\delta{{H}}_{1}}{\delta{n}}}
    =\mathcal{E}_{2}\binom{\frac{\delta {H}_{2}}{\delta{m}}}{\frac{\delta{{H}}_{2}}{\delta{n}}},\quad {H}_{i}= \frac{1}{2}\int h_i dx,\quad (i=1, 2),
     \end{equation*}
 where
\begin{equation*}
    \begin{aligned}
	h_{1}=&\left(v_{x}+2 \beta v-u\right)\left(n_{x}+2 \beta n+m\right)-n^{2}, \\ 
	h_{2}=&\left(u-v_{x}-2 \beta v\right)\left[\left(v_{x x}+2 \beta v_{x}-u_{x}\right)^{2}+4\left(\alpha+\beta^{2}\right)\left(v_{x}+2 \beta v-u\right)^{2}-n^{2}\right].
	     \end{aligned}
\end{equation*} 
\section{Connection to the two-component CH equation}
We now make the following scaling
\[
\partial_x\to \kappa\partial_x, \quad \partial_t\to 4\kappa\partial_t, \quad m\to \kappa^2 m, \quad n\to \kappa^2 n \quad (\kappa^2=4(\alpha+\beta^2))
\]
so that Eq.  \eqref{eq1} becomes
\begin{align}\label{eq1-1}
    \left\{     
        \begin{aligned}
    m_{ t}&=(\partial_x m+m \partial_x)\left(\kappa v_{ x}+2 \beta v-u\right)+ 
    n \left(\kappa u_{x}-2 \beta u-4 \alpha v \right)_x, \\
   n_{ t}&=  \big[n \left(\kappa v_{ x}+2 \beta v-u\right)\big]_x ,\\
   m&=u_{x x}-u, \quad   n=v-v_{ x x},
\end{aligned}
\right.
\end{align} 
and we take 
\begin{equation*}
    \mathcal{E}_{1} =\left(\begin{array}{cc}
       m \partial_{x}+\partial_{x} m& n \partial_{x}\\
       \partial_{x} n&0
       \end{array}\right),\quad \mathcal{E}_{2}=\frac{1}{\kappa^2} \left(\begin{array}{cc}
             -4\alpha \partial_{x}&-\kappa\partial_{x}^{2}-2\beta \partial_{x}\\
        \kappa\partial_{x}^{2}-2\beta \partial_{x}& \partial_{x}  
     \end{array}\right),
\end{equation*} 
as its Hamiltonian operators.

Let us   recall the two-component CH equation  \cite{clz2006,olver1996}. It is given by 
\begin{equation}
    \begin{array}{l}\label{key2-1}
        \left\{\begin{array}{l}
        p_{t}+q p_{x}+2 p q_{x}-\rho \rho_{x}=0 ,\\
        \rho_{t}+\left(\rho q\right)_{x}=0 ,\\
        p=q-q_{xx},
        \end{array}\right.
    \end{array}
    \end{equation}
which is also a bi-Hamiltonian system, namely
 \begin{equation*}
         \begin{pmatrix}p\\ \rho\end{pmatrix}_t=\mathcal{B}_{1} \binom{\frac{\delta {H}_{1}}{\delta{p}}}{\frac{\delta{{H}}_{2}}{\delta{\rho}}}
   =\mathcal{B}_{2}\binom{\frac{\delta {H}_{2}}{\delta{p}}}{\frac{\delta{{H}}_{2}}{\delta{\rho}}}, 
     \end{equation*}
  where   
  \begin{align*}
   \mathcal{B}_{1}=\left(\begin{array}{cc}
         -p\partial_{x}-\partial_{x} p & -\rho \partial_{x} \\
          -\partial_{x} \rho & 0\\
         \end{array}\right),\quad
         \mathcal{B}_{2}=\left(\begin{array}{cc}
          \partial_{x}^{3}-\partial_{x} & 0 \\
          0 &  \partial_{x}
       \end{array}\right)
 \end{align*}
 and 
 \begin{align*}
  {H}_{1} = \frac{1}{2}\int  \left( q p-\rho^{2}\right)dx,\quad
  {H}_{2}= \frac{1}{2}\int q\left(q^2+q_x^2-\rho^{2}\right) dx.\end{align*}
are Hamiltonian functionals.

It is interesting that Eq. \eqref{eq1-1} is  related to Eq. \eqref{key2-1} and the result is summarized as follows.
\begin{prop}
The transformation defined by 
\begin{align}\label{trans1}
      \left\{\begin{aligned}
      p&=-(\kappa{n} _{{x}}+2 \beta {n}+{m}),\\
	\rho&=\kappa{n}, \\
	  q&=-\kappa v_{x}-2\beta v+u, 
      \end{aligned}\right.
 \end{align}
 is a Miura type transformation between \eqref{eq1-1} and \eqref{key2-1}.
\end{prop}
The proposition, which may be proved  by a direct calculation, shows that system \eqref{eq1-1} is a modification of \eqref{key2-1}. Furthermore, we can easily check the following 
\begin{equation*}
\mathcal{B}_1=\mathcal{M}\mathcal{E}_1 \mathcal{M}^{\dagger}, \quad
\mathcal{B}_2=\mathcal{M} \mathcal{E}_2 \mathcal{M}^{\dagger},
\end{equation*} 
where 
\begin{equation*}
    \mathcal{M}=\left(\begin{array}{cc} 
   -1&-\kappa \partial_x-2\beta\\
   0& \kappa
    \end{array}\right)
\end{equation*}
is the Frechet derivative of the mapping defined by the first two equations of \eqref{trans1}. Therefore, as $\mathcal{B}_1$ and $\mathcal{B}_2$ are compatible Hamiltonian operators, $\mathcal{E}_1$ and $ \mathcal{E}_2$ have the same property.

\section{Reductions}
      Eq. \eqref{eq1} is a two-component system, so it is interesting to consider its reductions. When $v=0$ it reduces to the Camassa-Holm equation \eqref{eq0-0}, which has been studied extensively. For $u=0$ and $\alpha=0$, one obtains Eq. \eqref{eq01} 
which was considered as a new peakon equation by Xue, Du and Geng \cite{xdg2023}. Indeed, they also showed that this equation admits $N$-peakon solution.

To get a better understanding of \eqref{eq01}, we consider 
\begin{equation}\label{eq17}
n_{t}=\left(-\big(\frac{1}{n}\big)_{x}+2 \beta \frac{1}{n}\right)_{x},
\end{equation}
which shares the same hierarchy with \eqref{eq01}. To see it, we notice that Eq. \eqref{eq17} has the following spectral problem
\begin{align*}\label{eq16}
   \binom{\phi_{1}}{\phi_{2}}_{x}&=\left(\begin{array}{cc}
        \beta+\lambda n & 1 \\
        0 & -\beta-\lambda n
        \end{array}\right)\binom{\phi_{1}}{\phi_{2}},\\
\binom{\phi_{1}}{\phi_{2}}_{t}&=\left(\begin{array}{cc}
       2\lambda^2+\left(\frac{2\beta}{n}-(\frac{1}{n})_x\right)\lambda& \frac{2}{n} \lambda \\
        0 & - 2\lambda^2-\left(\frac{2\beta}{n}-(\frac{1}{n})_x\right)\lambda    
\end{array}\right)\binom{\phi_{1}}{\phi_{2}}.
\end{align*}

%
%
Eq. \eqref{eq17} is in conservation form, so we define the reciprocal transformation $(x, t, n(x,t)) \to (y, \tau, s(y,\tau))$ defined by
\begin{equation}\label{reci}
\left\{\begin{aligned}
n(x,t)&=s(y,\tau),\\
\partial_{x}&=n\partial _{y},\\
\partial_{ t}&=\left(-n\left(\frac{1}{n}\right)_{y}+2 \beta \frac{1}{n}\right) 	\partial _{y}+\partial _{\tau}.
\end{aligned}\right.\end{equation}
This transformation brings Eq.(\ref{eq17})  to 
\begin{equation}\label{eq20}s_{\tau}=s_{y y}- \frac{2}{s}\left(s_{y}^{2}+2  \beta s_{y}\right),
\end{equation}
which, after setting  $\displaystyle
s=\frac{1}{\omega}$, becomes the celebrated Burgers equation
\begin{equation}\label{eq18}
\omega_{\tau}=\omega_{y y}-4\beta \omega \omega_{y}.
\end{equation}
It is well known that the Burgers equation possesses the following recursion operator
\begin{equation*}R_{\omega}=
\partial_{y}-2 \beta \omega-2 \beta \omega _{y} \partial_{y}^{-1}. 
\end{equation*}
So we obtain a recursion operator 
\begin{equation*}R_{s}=\omega^{-1} sR_{\omega}s^{-1}
 \omega=s^{2}\left(\partial_{y}-2 \beta s^{-1}+2 \beta s^{-2} s_{y} \partial_{y}^{-1}\right)s^{-2}
\end{equation*}
for Eq.\eqref{eq20}. 

In addition, the reciprocal transformation \eqref{reci} between Eqs.(\ref{eq17}) and (\ref{eq20}) may  be formulated as a B\"acklund  transformation
\begin{equation*}
B(n,s)=n(x)-s((\partial_{x}^{-1} n)(x))=0.
\end{equation*}
Thus, from the recursion operator $R_s$ for Eq.(\ref{eq20}), we find a recursion operator for Eq.(\ref{eq17}), namely
\begin{equation*}R_{n}=\partial_{x} n \partial_{x}^{-1} n^{-1} R_{s} n \partial_{x} n^{-1} \partial_{x}^{-1}=\partial_{x}(\partial_{x}-2 \beta) n^{-1} \partial_{x}^{-1}.
\end{equation*}
It is interesting to observe that in terms of above operator Eqs.(\ref{eq01}) and (\ref{eq17}) can be reformulated as 
\begin{equation*}
n_{t }=4R_{n}^{-1} n_{x},\quad R_{n}^{-1} n_{t }=0,
\end{equation*} 
respectively. 
In conclusion, we have shown that both Eq. \eqref{eq01} and Eq. \eqref{eq17} belong to a hierarchy which is reciprocally related to the Burgers hierarchy.

\section*{Acknowledgments}

This work is supported by the National Natural Science Foundation of China (Grant Nos. 11931107 and 12171474).
\bibliographystyle{unsrt}
\bibliography{XDGbib01}

\begin{thebibliography}{10}

\bibitem{CH1993}
Roberto Camassa and Darryl~D Holm.
\newblock An integrable shallow water equation with peaked solitons.
\newblock {\em Phys. Rev. Lett.}, 71(11):1661--1664, 1993.

\bibitem{Fuchss96}
Benno Fuchssteiner.
\newblock Some tricks from the symmetry-toolbox for nonlinear equations:
  generalizations of the {C}amassa-{H}olm equation.
\newblock {\em Phys. D}, 95(3-4):229--243, 1996.

\bibitem{Schiff1998}
Jeremy Schiff.
\newblock The {C}amassa-{H}olm equation: a loop group approach.
\newblock {\em Phys. D}, 121(1-2):24--43, 1998.

\bibitem{BSS2000}
Richard Beals, David~H Sattinger, and Jacek Szmigielski.
\newblock Multipeakons and the classical moment problem.
\newblock {\em Adv. Math.}, 154(2):229--257, 2000.

\bibitem{Johnson2002}
Robin~S Johnson.
\newblock {Camassa-Holm, Korteweg-de Vries and related models for water waves}.
\newblock {\em J. Fluid Mech.}, 455:63--82, 2002.

\bibitem{LZ2004}
Yishen Li and Jin~E. Zhang.
\newblock The multiple-soliton solution of the {C}amassa-{H}olm equation.
\newblock {\em Proc. R. Soc. Lond. Ser. A}, 460(2049):2617--2627, 2004.

\bibitem{Parker2005}
Allen Parker.
\newblock On the {C}amassa-{H}olm equation and a direct method of solution.
  {III}. {$N$}-soliton solutions.
\newblock {\em Proc. R. Soc. Lond. Ser. A}, 461(2064):3893--3911, 2005.

\bibitem{CGI2006}
Adrian Constantin, Vladimir~S Gerdjikov, and Rossen~I Ivanov.
\newblock Inverse scattering transform for the {C}amassa-{H}olm equation.
\newblock {\em Inverse Problems}, 22(6):2197--2207, 2006.

\bibitem{Monvel2007}
Anne Boutet~de Monvel and Dmitry Shepelsky.
\newblock Riemann-{H}ilbert approach for the {C}amassa-{H}olm equation on the
  line.
\newblock {\em C. R. Math. Acad. Sci. Paris}, 343(10):627--632, 2006.

\bibitem{CL2009}
Adrian Constantin and David Lannes.
\newblock The hydrodynamical relevance of {the Camassa-Holm and
  Degasperis-Procesi equations}.
\newblock {\em Arch. Ration. Mech. An.}, 192(1):165--186, 2009.

\bibitem{Feng10}
Bao-Feng Feng, Ken-ichi Maruno, and Yasuhiro Ohta.
\newblock A self-adaptive moving mesh method for the {C}amassa-{H}olm equation.
\newblock {\em J. Comput. Appl. Math.}, 235(1):229--243, 2010.

\bibitem{Xia16}
Baoqiang Xia, Ruguang Zhou, and Zhijun Qiao.
\newblock Darboux transformation and multi-soliton solutions of the
  {C}amassa-{H}olm equation and modified {C}amassa-{H}olm equation.
\newblock {\em J. Math. Phys.}, 57(10):103502, 2016.

\bibitem{Rasin-Schiff2017}
Alexander~G Rasin and Jeremy Schiff.
\newblock {B}äcklund transformations for the {Camassa–Holm} equation.
\newblock {\em J. Nonlinear Sci.}, 27:45--69, 2017.

\bibitem{Li-Tian}
Nianhua Li and Kai Tian.
\newblock {Nonlocal symmetries and Darboux transformations of the
  Camassa–Holm equation and modified Camassa–Holm equation revisited}.
\newblock {\em J. Math. Phys.}, 63(4):041501, 2022.

\bibitem{Chang2024}
Xiang-Ke Chang and Xiao-Min Chen.
\newblock On the peakon dynamical system of the second flow in the
  {Camassa–Holm} hierarchy.
\newblock {\em Adv. Math.}, 459:110000, 2024.
\newblock Cited by: 0.

\bibitem{DP}
Antonio Degasperis and Michela. Procesi.
\newblock Asymptotic integrability.
\newblock In {\em Symmetry and perturbation theory ({R}ome, 1998)}, pages
  23--37. World Sci. Publ., River Edge, NJ, 1999.

\bibitem{Novikov}
Vladimir Novikov.
\newblock Generalizations of the {C}amassa-{H}olm equation.
\newblock {\em J. Phys. A: Math. Theor.}, 42(34):342002, 2009.

\bibitem{olver1996}
Peter~J Olver and Philip Rosenau.
\newblock {Tri-Hamiltonian duality between solitons and solitary-wave solutions
  having compact support}.
\newblock {\em Phys. Rev. E}, 53(2):1900--1906, 1996.

\bibitem{Qiao}
Zhijun Qiao.
\newblock A new integrable equation with cuspons and {W/M}-shape-peaks
  solitons.
\newblock {\em J. Math. Phys.}, 47(11):112701, 2006.

\bibitem{clz2006}
Ming Chen, Si-Qi Liu, and Youjin Zhang.
\newblock A two-component generalization of the {C}amassa-{H}olm equation and
  its solutions.
\newblock {\em Lett. Math. Phys.}, 75:1--15, 2006.

\bibitem{Wu2006}
Chao-Zhong Wu.
\newblock On solutions of the two-component {Camassa-Holm system}.
\newblock {\em J. Math. Phys.}, 47(8):083513, 2006.
\newblock Cited by: 11.

\bibitem{HI2011}
Darryl~D Holm and Rossen~I Ivanov.
\newblock Two-component {CH} system: {Inverse Scattering, Peakons and
  Geometry}.
\newblock {\em Inverse Problems}, 27(4):45013--45031, 2011.

\bibitem{Matsuno2017}
Yoshimasa Matsuno.
\newblock Multisoliton solutions of the two-component {Camassa-Holm} system and
  their reductions.
\newblock {\em J. Phys. A: Math. Theor.}, 50(34):345202, 2017.
\newblock Cited by: 7; All Open Access, Green Open Access.

\bibitem{Wang2020}
Gaihua Wang, Nianhua Li, and Q.P. Liu.
\newblock {Multi-soliton solutions of a two-component Camassa-Holm system:
  Darboux} transformation approach.
\newblock {\em Commun. Theor. Phys.}, 72(4):045003, 2020.
\newblock Cited by: 4.

\bibitem{Qu2024}
Cheng He, Xiaochuan Liu, and Changzheng Qu.
\newblock Orbital stability of two-component peakons.
\newblock {\em Sci. China Math.}, 66(7):1395–--1428, 2023.
\newblock Cited by: 1; All Open Access, Bronze Open Access, Green Open Access.

\bibitem{hnw2017}
Andrew~NW Hone, Vladimir Novikov, and Jing~Ping Wang.
\newblock Two-component generalizations of the {Camassa--Holm} equation.
\newblock {\em Nonlinearity}, 30(2):622--658, 2017.

\bibitem{xdg2023}
Bo~Xue, Huiling Du, and Xianguo Geng.
\newblock New integrable peakon equation and its dynamic system.
\newblock {\em Appl. Math. Lett.}, 146:108795, 2023.

\end{thebibliography}

\end{document}